\documentclass[journal]{IEEEtran}

\ifCLASSINFOpdf
\else
   \usepackage[dvips]{graphicx}
\fi
\usepackage{url}
\usepackage{amsmath}
\usepackage{graphicx}
\usepackage{placeins}
\usepackage{float}
\usepackage{bbm}
\usepackage{subfigure}
\usepackage{verbatim}
\usepackage{amssymb}
\usepackage{amsmath}
\usepackage{amsthm}
\usepackage{mwe}
\usepackage{docmute}
\usepackage{mathtools}
\usepackage[ruled, vlined]{algorithm2e}
\usepackage{bm}
\usepackage{siunitx}

\usepackage[disable]{todonotes}

\usepackage{graphicx}

\DeclareMathOperator*{\argmin}{argmin}

\newtheorem{prop}{Proposition}
\newtheorem{definition}{Definition}

\theoremstyle{definition}

\newtheorem{assumption}{Assumption}

\begin{document}

\title{Mode Selection and Target Classification in Cognitive Radar Networks}
% \author{William W. Howard, Anthony F. Martone, R. Michael Buehrer
\author{William W. Howard, Samuel R. Shebert, Benjamin H. Kirk, R. Michael Buehrer
\thanks{W.W. Howard, S.R. Shebert and R.M. Buehrer are with Wireless@VT, Bradley Department of ECE, Virginia Tech, Blacksburg, VA, 24061. \\ 
B.H. Kirk is with the U.S. Army Research Laboratory, Adelphi, MD 20783. (e-mail:{benjamin.h.kirk.civ}@army.mil).\\
Distribution Statement A: Approved for public release. Distribution is unlimited. \\
Contact:$\{${wwhoward, sshebert}$\}$@vt.edu  \\
}
}

\maketitle
\pagenumbering{gobble}

%%%%%%%%%%%%%%%%%%%%%%%%%%%%%%%%%%%%%%%%%%%%%%%%%%%%%%%%%%%%%%%%%%%%%%%%%%%%%%%%%%%%%%%%%
\begin{abstract}
Cognitive Radar Networks were proposed by Simon Haykin in 2006 to address problems with large legacy radar implementations - primarily, single-point vulnerabilities and lack of adaptability. 
This work proposes to leverage the adaptability of cognitive radar networks to trade between active radar observation, which uses high power and risks interception, and passive signal parameter estimation, which uses target emissions to gain side information and lower the power necessary to accurately track multiple targets. 
The goal of the network is to learn over many target tracks both the characteristics of the targets as well as the optimal action choices for each type of target.
In order to select between the available actions, we utilize a multi-armed bandit model, using current class information as prior information. 
When the active radar action is selected, the node estimates the physical behavior of targets through the radar emissions.
When the passive action is selected, the node estimates the radio behavior of targets through passive sensing. 
Over many target tracks, the network collects the observed behavior of targets and forms clusters of similarly-behaved targets. 
In this way, the network meta-learns the target class distributions while learning the optimal mode selections for each target class. 
\end{abstract}


\section{Introduction}
\label{sec:intro}
There is a desire for low size, weight, and power (\textbf{SWaP}) devices that must operate in diverse environments. 
To this end, we examine a network of low-power, multi-mode cognitive radar devices (Cognitive Radar Network, or \textbf{CRN}) \cite{howard2022_MMABjournal} \cite{haykin2005}. 
These radar devices (\emph{nodes}) \cite{thornton2022_universaljournal} can choose between active radar observation and passive signal detection and classification (Electronic Support Measures, \textbf{ESM}) in each of many time steps\footnote{Active and passive measurement are used to detect and track multiple objects. } \cite{Bogler1987}. 
Over the course of many target tracks, the CRN (via a central coordinator, or \textbf{CC}) must form \emph{classes} of targets which share similar properties. 
We show that if the class of a given target determines its physical behavior, more accurate knowledge of the classes in an environment will 1) reduce the effective radiated power from each node and thereby the CRN's probability of being detected, and 2) improve target tracking performance. 

To provide an alternative method (i.e., other than radar) of observing targets, we leverage the fact that most modern targets, civilian or military, tend to have characteristic radio emissions (e.g., FM voice communication and Automatic Dependent Surveillance-Broadcast (\textbf{ADS-B}) in general aviation; control, telemetry, and data downlink in consumer unmanned aerial vehicles (\textbf{UAVs}); two-way voice communications in aerial balloons). 
By sensing the spectrum rather than conducting radar observations, we allow the nodes to take advantage of additional information while reducing the operating power requirements. 
In addition, these different broad classes of targets exhibit different physical motion. Generally, commercial aircraft travel at moderate to high velocities and conduct low-G maneuvering, consumer drones have slow to moderate velocities and with the capability for high-G maneuvers, and balloons have slow velocities and acceleration. 
Therefore, accurately estimating the radio behavior of these targets will map strongly to specific physical behavior. 
We will further develop the idea of target classes in Section \ref{sec:targets}. 

By forming models of these target classes over time, we show that the CC can selectively choose between active and passive modes for each node to quickly identify the class of a target and allocate the necessary radar time to ensure adequate tracking performance. 
We discuss a Multi-Armed Bandit (\textbf{MAB}) model for mode selection, where in each time step the CC evaluates the scenario to update the bandit model and select new modes for the next time step. 
While in general the scenario need not be stationary with respect to reward, we make the assumption that each epoch is sufficiently short in time and the scenario sufficiently large in space such that targets do not move far from their original position.  

The CRN we discuss in this work contains a single cognitive CC which selects actions for $N$ multi-mode nodes\footnote{We assume a stochastic Poisson point process (\textbf{PPP}) generates the locations for these nodes, and say that $N$ is the random variable described by this model. Similarly, $M$ is the random variable describing the number of targets. }. 
Similarly, there are $M$ targets throughout the environment with physical and radio characteristics that we will discuss in Section \ref{sec:targets}.

\subsection{Contributions} 
We build on previous contributions in the areas of radar network control and cognitive radar. Specifically, we contribute the following: 
\begin{itemize}
    \item A model for mode selection in multi-function sensor node networks (specifically CRNs). 
    \item An analysis of multiple target classes and their characteristic motion and signal models. 
    \item Mathematical analysis on a technique to cluster targets by class. 
    \item Numerical simulations to support our conclusions.  
    \item We show that our proposed technique outperforms pure radar observation as well outperforming random mode selection without class formation. 
\end{itemize}

\subsection{Notation}
We use the following notation. 
Matrices and vectors are denoted as bold upper $\mathbf{X}$ or lower $\mathbf{x}$ case letters respectively.
% Element-wise multiplication of two matrices or vectors is shown as $\mathbf{X}\odot \mathbf{Y}$. 
Functions are shown as plain letters $F$ or $f$. 
Sets $\mathcal{A}$ are shown as script letters. 
% The cardinality $|\mathcal{A}|$ of a set $\mathcal{A}$ refers to the number of elements in that set. 
% The logical negation of a statement $a$ is given by an overline $\overline{a}$. 
The transpose operation is $\mathbf{X}^T$. 
% The backslash $\mathcal{A}\backslash \mathcal{B}$ represents the set difference. 
% Boxes (intervals) in $\mathbb{R}^d$ are written as $[a,b]^d$ and when the elements of a set are denoted, they are given as $\mathcal{A} = \{a, b, c, \dots \}$. 
% Functions which indicate whether a value $x$ is present in a set $\mathcal{A}$ are denoted as $\mathbbm{1}_\mathcal{A}(x)$. 
Random variables are written as upper-case letters $X$, and their distributions will be specified. 
The set of all real numbers is $\mathbb{R}$ and the set of integers is $\mathbb{Z}$. 
% The speed of electromagnetic radiation in a vacuum is given as $c$. 
% Element-wise (scalar) multiplication is denoted as $x*y$. 
The Euclidean norm of a vector $\mathbf{x}$ is written as $||\mathbf{x}||$. 
Estimates of a true parameter $p$ are given as $\hat{p}$. 

\subsection{Organization} 
Section \ref{sec:background} covers recent work in related areas and formulates the problem. 
Section \ref{sec:targets} provides a mathematical description of targets. 
Section \ref{sec:methods} discusses the methods we implement to solve the problem. 
Section \ref{sec:simulations} provides numerical simulations to support our conclusions which are drawn in Section \ref{sec:conclusions}.

%%%%%%%%%%%%%%%%%%%%%%%%%%%%%%%%%%%%%%%%%%%%%%%%%%%%%%%%%%%%%%%%%%%%%%%%%%%%%%%%%%%%%%%%%
%%%%%%%%%%%%%%%%%%%%%%%%%%%%%%%%%%%%%%%%%%%%%%%%%%%%%%%%%%%%%%%%%%%%%%%%%%%%%%%%%%%%%%%%%

\section{Background}
\label{sec:background}
CRNs and general cognitive radar are two areas of recent study. 
Recent contributions regarding cognitive radar include  studies of waveform selection and optimization \cite{thornton2022_universaljournal} \cite{thornton2020efficient} and meta-cognition \cite{martone2017adaptable}. 
Regarding CRNs, previous works have investigated collaborative learning \cite{howard2022_MMABjournal} \cite{howard2022_decentralized_conf} \cite{howard2023_hybridjournal} and Age-of-Information metrics  \cite{howard2023_timelyjournal}. 

% Tmp lit search
The authors of \cite{6807568} present a centralized passive estimation network, where the nodes act as amplify-and-forward unit, where the total amplification of the network is power-limited. 
All of the decision-making in the network is located in the central coordinator, and is confined to allocating the limited amplification power to the nodes. 
The targets are modeled as having a constant quadratic mean complex-valued signal. 
The FC fuses the node measurements with a goal of estimating the true target signal. 
This work is useful because it provides a framework for signal estimation, which could support classification. 
However, since the different measurements are fused, this technique is not useful for direction of arrival estimation. 

The problem of multi-sensor target identification has been previously studied in \cite{Bogler1987, Hong1993} and more recently in \cite{Lei2020, Li2022}. These works use Dempster-Shafer evidence theory to fuse measurements from multiple sensors to improve target identification rates. In \cite{Challa2001} and \cite{Cao2018}, radar and passive signal classification are combined to improve tracking performance by classifying targets. However, it is assumed that the pairing between radar targets and emitted signals are known \textit{a priori}.

\section{Target Modeling}
\label{sec:targets}
\subsection{Class Definitions}
\label{ss:classes}
Generally, each target in the observable region $R$ can be partially or wholly described by several parameters. 
Let the parameters describing a target $m$ be collected into an ordered set $\mathcal{X}_m$. 
We describe a parameter $E_m$ of $\mathcal{X}_m$ as a Markov process with a stationary distribution $\boldsymbol{\pi}(E_m)=[\pi_1, \pi_2, \dots, \pi_p]$ and a state transition matrix $\mathbf{P}(E_m)$. 
The stationary distribution describes the probability that $E$ is\footnote{Sometimes the subscript is omitted in order to describe property $E$ without dependence on the particular target. } in each of the $p$ states, and the state transition matrix consists of entries $p_{i,j}$ which describe the probability that $E$ switches from state $i$ to state $j$. 

We begin with a few definitions: 

\begin{definition}[Equal in State Distribution]
\label{def:eq_in_dist}
    Two random variables $X$ and $Y$ are said to be equal in state distribution if they consist of the same number of states, and if the stationary distribution is equal.  
\end{definition}

\begin{definition}[Target Class]
\label{def:target_class}
    Let target $m_0$ be described by $n_0$ observable parameters collected into the set $\mathcal{X}_0$. If a target $m_1$ has the property that each parameter in $\mathcal{X}_1$ is equal in state distribution with the corresponding property in $\mathcal{X}_0$, then target $m_0$ is said to be of the same class as $m_1$.  Denote the class as $\mathcal{C}$. 
\end{definition}

\begin{definition}[Target Family]
\label{def:target_family}
    A family $\mathcal{F}$ is a group of target classes $\{\mathcal{C}_1, \mathcal{C}_2, \dots\}$ with the following properties: 
\begin{itemize}
	\item Each of $\mathcal{C}_1, \mathcal{C}_2, \dots$ consists of the same number of parameters, and each group of parameters shares the same number of states\footnote{In other words, if one class $\mathcal{C}_i\in\mathcal{F}$ consists of parameters $\mathcal{E}_1, \mathcal{E}_2, \mathcal{E}_3$ which have $p_1, p_2, p_3$ states respectively, then all classes $\mathcal{C}_{j\neq i} \in \mathcal{F}$ will also consist of those parameters and states. }. 
	\item Within a family, there is a one-to-one and onto mapping between the parameters of a given class. 
\end{itemize}
\end{definition}

\begin{prop}[Unique Class]
\label{prop:unique}
Let target $m$ with parameter estimate $\hat{\mathcal{X}}_m$ be drawn from a family $\mathcal{F}$. Then, the parameter estimate $\hat{\mathcal{X}}_M$ maps to a unique class. 
\end{prop}

\begin{proof}[Proof of Prop. \ref{prop:unique}]
By Def. \ref{def:target_family}, and since functions which are one-to-one and onto (i.e., bijective) are invertible, there is a uniquely identifiable class $\mathcal{C}\in\mathcal{F}$ for a parameter estimate $\hat{\mathcal{X}}_m$\footnote{Note however that since this parameter estimate is not necessarily correct, the estimated class is not necessarily correct. }. 
\end{proof}

Targets within a family $\mathcal{F}$ have the useful feature that if the stationary distribution of one element $E_m$ of target $\mathcal{X}_m$ can be estimated, the other elements can be identified, since no other classes in the family will have the same stationary distribution. 
Said another way, the elements of $\mathcal{X}_m$ which is of $\mathcal{C}_i$ for any class $i$ in a family $\mathcal{F}$ are well characterized if one element is known (see the proof of Prop \ref{prop:unique}).

\begin{table}
    \centering
    \caption{Example Target Classes}
    \begin{tabular}{c|c|c}
        %\multicolumn{3}{c}{Simulation Parameters} \\
        Example Class & Motion & Signals \\[0.5ex] 
        \hline
        UAV & Mid-altitude, dynamic & Control, telemetry \\
        General Aviation & Low-altitude, stable & ADS-B, FM Voice \\
        High-Altitude Balloon & Low maneuverability & Telemetry \\ [1ex] 
        % \hline
    \end{tabular}    
    \label{tab:classes}
\end{table}

\subsection{Spectrum Sensing}
\label{ss:sensing}
The targets of interest operate their own radar and/or communication systems, resulting in electromagnetic emissions that can be detected by a spectrum sensing receiver.
A key aspect of passive spectrum sensing is that a target's transmitter characteristics impact the probability of detection at the receiver, even if the target is nearby. 
This is contrary to radar, which can detect in-range targets in a 'non-cooperative' manner.

The detection performance of the spectrum sensing receiver is modeled in two parts: the maximum detectable range and the probability of intercept of the target's emissions. 
\subsubsection{Maximum Detectable Range}
The maximum detectable range of a target depends on the transmit power ($P_t$), the transmit and receive antenna gains ($G_t$ and $G_r$), the wavelength ($\lambda$), the distance between the target and sensor ($R$), the receivers noise power ($P_n$), and the transmit and receiver losses ($L$). The instantaneous (i.e., non-integrated) signal to noise ratio (SNR) at the spectrum sensing receiver for airborne targets will be \cite[Chap. 2]{wiley2006}: 
\begin{equation}
	SNR = \frac{P_t G_t G_r \lambda^2}{(4 \pi R)^2 P_n L}
\end{equation}
Where the receiver noise power is:
\begin{equation}
	P_n = k T_0 F B
\end{equation}
$k$ is Boltzmann's constant, $T_0$ is 290 Kelvin, $F$ is the receiver noise figure, and $B$ is the receiver bandwidth. We assume a noise figure of 10 and the receive bandwidth is 1 MHz per channel. The SNR as a function of range and transmit power is shown in Figure \ref{fig:receiver_snr}. 

The SNR required for detection will depend on the detection algorithm, signal capture duration, and the signal of interest. For cyclostationary detectors with a signal duration greater than a few milliseconds, it is possible to detect signals at a rate close to 100\% at or below 0 dB SNR with a false alarm rate less than 1/100 \cite{cyclo_ofdm_detection, cyclo_lte_detection}. Therefore, we assume that an SNR of 0 dB provides nearly 100\% detection rate and consider the false alarm rate to be negligible. Further, we assume that each signal of interest has a unique set of \textit{cycle frequencies} that can be identified by cyclostationary detectors \cite{cyclo_ofdm_detection, cyclo_lte_detection} with negligible confusion between signal classes.

The maximum detectable range for passive detection of a target is defined as the maximum range that 0 dB SNR is achieved. Based on Figure \ref{fig:receiver_snr}, the maximum range ranges from 10s to 100s of km depending on the target's transmit power. Targets that are within the maximum detectable range are referred to as \textit{in-range}.

\begin{figure}
    \centering
    \includegraphics[scale=0.6]{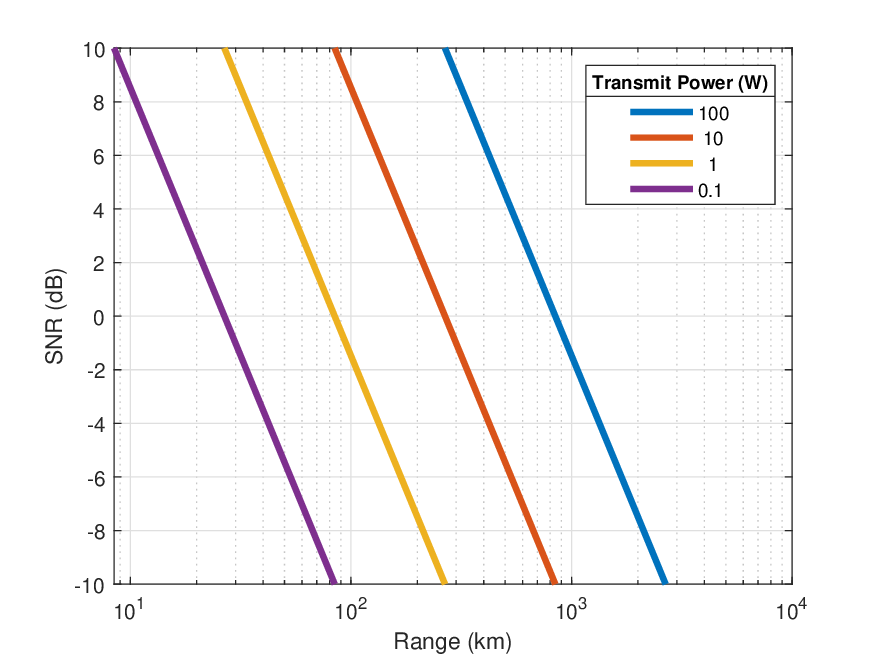}
    \caption{Spectrum sensing receiver SNR at a center frequency of 1 GHz, omnidirectional transmit and receive antennas, and 3 dB of losses. Given a SNR requirement of 0 dB for detection, targets can reasonably be detected 10s to 100s of km away.}
    \label{fig:receiver_snr}
\end{figure}

\subsubsection{Probability of Intercept}
The probability of intercept characterizes dynamic factors that decrease the rate that an in-range target is detected.
In general, the probability of intercept depends on the sensing receivers measurements in time, frequency, and space overlapping with the targets transmit time, frequency, and direction \cite[Chap. 4]{wiley2006}. For simplicity, we assume that the receiver has wide enough bandwidth to observe all target emissions in a time step, and both transmitter and receiver use omnidirectional antennas\footnote{Note that using directional antennas would increase the maximum detectable range, but decrease the probability of intercept.}. Therefore, the probability of intercept is governed by the transmit fraction of time. The on/off characteristics of $m^{th}$ transmitter is modeled as a two state Markov chain with transition probability matrix $P(tx)$:
\begin{equation}
	P(tx) = \begin{bmatrix} \text{Pr}(\text{On} | \text{On}) &  \text{Pr}(\text{Off} | \text{On})\\ \text{Pr}(\text{On} | \text{Off}) &  \text{Pr}(\text{Off} | \text{Off})\end{bmatrix}
\end{equation}

Further, the signal type emitted by the $m^{th}$ transmitter is one of $s$ finite states with stationary distribution $\pi(S_m)$ and transition matrix $P(S_m)$. 
%Similar to previous work \cite{Challa2001, Cao2018}, we assume that the probability of intercept for the targets of interest is high, such that spectrum sensing is effective.
Therefore, a target is passively detected by a sensing node when the target is \textit{in-range}, the target is actively transmitting, and the sensing node is in spectrum sensing mode (as opposed to radar mode) for the given time step.

%\begin{figure}
%    \centering
%    \includegraphics[scale=0.6]{figs/cyclo_detector_perf.eps}
%    \caption{Performance of cyclostationary detector \cite{cyclo_detector} on LTE downlink signals}
%    \label{fig:cyclo_detector}
%\end{figure}

\section{Methods}
\label{sec:methods}
Let $M$ targets $\{m_1, m_2, \dots\}$ exist in the region $R$ with the position of each target $m$ described by $\textbf{X}_m(t)=[x_m(t), y_m(t), z_m(t)]$. 
The function $X_m(t)$ evolves in time according to a Markov motion model $V_m(t)$ \cite{ristic2003beyond} which describes the type of motion with stationary distribution $\boldsymbol{\pi}(V_m)=[\pi_1,\pi_2,\dots,\pi_p]$ representing the motion states and transition matrix $\mathbf{P}(V_m)$ representing the transition probabilities. 

Further, target $m$ emits a signal type with on/off characteristics that can be described by the Markov chain $S_m(t)$.
During time step $t_0$, the signal emitted by target $m$ are characterized by one of $s$ finite states. 
%A model has been pre-trained on each of these states, and can be used to estimate the signal parameters of the target. 
Coupled with a direction-of-arrival estimator, the signal classification estimate are associated with the radar targets.

So, the parameter set $\mathcal{X}$ for target $m$ consists of the random variable describing the motion  model $V_m(t)$ with $v$ states and the random variable describing the signal model $S_m(t)$ with $s$ states.

Over time and according to the modes selected for each node, the CC builds a set of estimated parameters for each observed target $m$, $\hat{\mathcal{X}}_m$. 
Using these parameters and the record of target behavior, the CC is able to tune filters to the expected motion model of the targets. This is shown in Figure \ref{fig:tuned_filters}, where the tracking accuracy using an untuned filter (i.e., no prior motion model known) and filter which is tuned to the correct motion model transition probabilities and process noise. 
The tuning is accomplished using the prior tracks of targets in each class. 
The performance is improved because the filter can use more accurate parameters earlier in the track. 

\begin{figure}[h]
    \centering
    \includegraphics[scale=0.6]{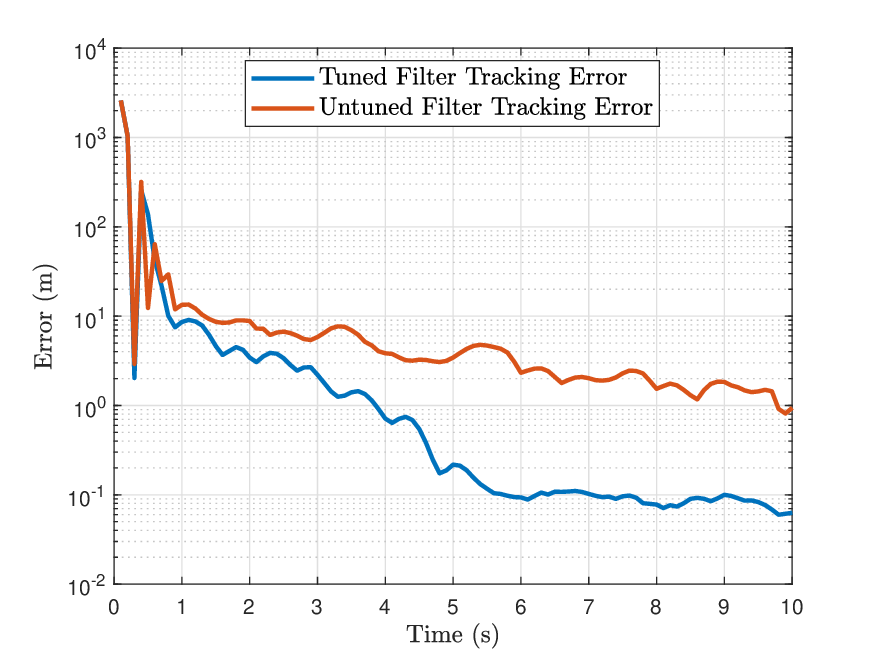}
    \caption{Kalman filters which are tuned to the process noise and motion model probabilities for a specific target class result in lower tracking error for targets of that class.  }
    \label{fig:tuned_filters}
\end{figure}

\subsection{Class Formation}
\begin{assumption}[Single Family]
The targets present in a given environment belong to a single family. 
\end{assumption}

When the game begins, the CC and radar nodes have no knowledge of the targets which are present. 
So the CC must, at several points in time, form and update target classes. 
This is accomplished through a k-means clustering process, using the mutual information between target parameter distributions as the distance. 
Targets which are sufficiently close are grouped together, and a class is formed by the mean of the parameter distributions. 
The number of target classes $\hat{k}$ is determined using the Akaike Information Criterion (\textbf{AIC}) \cite{1311138}. 

% When classes are formed and updated, 

\subsection{Target Measurement}
In order to measure the target parameter distributions, the radar nodes are equipped with two ``modes'': 
\begin{itemize}
\item \textbf{Active} measurement consists of radar processing for a single time step, and allows the radar node to estimate the position and velocity of a target by measuring the range, angle, radial velocity, and angular velocity of the target. 
\item \textbf{Passive} measurement consists of spectrum sensing for a single time step, and enables estimation of the signal class of an emission. Additionally, the direction of arrival is estimated to associate observed emissions to radar targets.
\end{itemize}

These two modes allow the FC to track over time the: 
\begin{itemize}
\item \textbf{Motion model} of the target, which consists of a Markov chain and several motion states. 
\item \textbf{Signal model} of the target, which consists of a Markov chain of signal types.
\end{itemize}

\subsection{Mode Selection}
Modes are selected once per time step for all nodes by the CC. 
Mode selection is accomplished by the evaluation of a utility matrix, which takes into account 1) the estimated target classes observed by each node, 2) the 
time since each node performed each action, and 3) the trade-offs inherent in multi-node networks. 
In other words, when a target is observed by two nodes and one of the nodes performs radar, there may be less interest in the second node also performing radar. 
We utilize a reinforcement learning agent to select the mode in each time step. 
Since each epoch is relatively short in duration, and the scene is spatially large compared to the distance a target can travel, each epoch can be assumed to be stationary in reward. 

\subsubsection{Learning Formulation}
We use the common Upper Confidence Bound (\textbf{UCB}) \cite{bandits} \cite{UCB_fischer} formulation, where a single player selects from finitely many actions (``arms'') and observes a corresponding reward. 
Over many iterations, the goal of the player is to maximize the total expected reward. 
To reduce the complexity\footnote{An alternative approach might assign a single bandit algorithm with one arm per combination of node actions, which would total $2^{N}$ arms. Our approach covers the same action space, while reducing the number of arms per bandit algorithm to two. }, we pose the problem with one bandit algorithm per node, which are all evaluated by the FC. 
The algorithms could possibly be implemented by each node, but since the reward function (shown below) requires global information, this approach would require more communication. 

\paragraph{Rewards} The reward for each action is generated by the normalized Shannon entropy of the motion model distribution. 
This value is used since it is constrained to the unit interval and reflects the information content of the motion model distribution: as the distribution of states becomes more flat, the Shannon entropy will increase. 
This is particularly useful because as targets become more maneuverable, Kalman filters become less accurate and therefore benefit from more frequent updating \cite{5977487}. 

\begin{equation}
\label{eq:rewards}
	u_n(t) = \frac1{M_n}\sum_{j=1}^{M_n}[\eta(V_j(t)), \eta(S_j(t))] 
\end{equation}

\begin{equation}
\label{eq:shannon}
	\eta(X(t)) = \sum_{i=1}^{n_X} \frac{x_i\log_2(x_i)}{\log_2(n_X)}
\end{equation}

Eq. (\ref{eq:rewards}) shows the reward for selecting either action at node $n$, where $\eta(\cdot)$ represents the normalized Shannon entropy. 
Eq. (\ref{eq:shannon}) shows the Shannon entropy for a distribution $X$ with $n_X$ states $x_i$.  
The reward for selecting the radar action is dependent on the distribution of motion states of covered targets, and the reward for selecting the passive action is dependent on the distribution of the signal states of covered targets. 
So, the reward formulation is dependent on the all of the targets viewed by a particular node. 

\paragraph{Mode Selection} Then, in each time step $t$, Eq. (\ref{eq:UCB}) is used to select the action for node $n$ where $N_t(n)$ is the number of times each mode has been selected before time $t$. 

\begin{equation}
\label{eq:UCB}
 \text{Mode}(t) = \argmin \left[u_n + \sqrt{\frac{\log t}{N_t(n)}} \right]
\end{equation}

%%%%%%%%%%%%%%%%%%%%%%%%%%%%%%%%%%%%%%%%%%%%%%%%%%%%%%%%%%%%%%%%%%%%%%%%%%%%%%%%%%%%%%%%%
%%%%%%%%%%%%%%%%%%%%%%%%%%%%%%%%%%%%%%%%%%%%%%%%%%%%%%%%%%%%%%%%%%%%%%%%%%%%%%%%%%%%%%%%%

%%%%%%%%%%%%%%%%%%%%%%%%%%%%%%%%%%%%%%%%%%%%%%%%%%%%%%%%%%%%%%%%%%%%%%%%%%%%%%%%%%%%%%%%%
%%%%%%%%%%%%%%%%%%%%%%%%%%%%%%%%%%%%%%%%%%%%%%%%%%%%%%%%%%%%%%%%%%%%%%%%%%%%%%%%%%%%%%%%%

\section{Numerical Simulations}
\label{sec:simulations}
We simulate a CRN with the parameters listed in Table \ref{tab:params}. 
In particular, we simulate fifteen epochs. 
After each epoch (i.e., scenario) the CC updates the list of target classes. 
When the game begins, there are no classes, and when it ends the classes should have high accuracy. 
Figure \ref{fig:classes} shows that this is the case; class accuracy increases in each epoch. 
Further, Figure \ref{fig:classes} also shows that the accuracy with which targets are associated to a class increases in each epoch. 
This, coupled with the result shown in Figure \ref{fig:tuned_filters} which shows that tracking accuracy improves when a tuned filter is used, implies that the observed tracking accuracy in the entire network should improve. 

\begin{table}
    \centering
    \caption{Simulation Parameters}
    \begin{tabular}{c|c|c}
        %\multicolumn{3}{c}{Simulation Parameters} \\
        Variable & Description & Value \\[0.5ex] 
        \hline
        $\lambda_n$ & Node Density per \qty{}{\km\squared} & 0.2 \\
        $\lambda_m$ & Target Density per \qty{}{\km\squared}& 0.3 \\
        $|B|$ & Simulated Region & \qty{100}{\km\squared} \\
         & Number of Classes & 3 \\
         & Averaged Simulations & 30 \\ 
         & Number of Epochs & 15 \\ 
         & Epoch Duration & 25s \\ [1ex]
        % \hline
    \end{tabular}    
    \label{tab:params}
\end{table}

Each epoch consists of 25 seconds of simulation, where target parameters evolve over time according to their motion and signal models. 
This time is long enough that the CC is able to estimate the transition probabilities for each target with sufficient accuracy, but not so long as to allow the targets to move very far from their initial positions (with respect to the size of the scenario, \qty{100}{\km\squared}). 
Future work will investigate the case where the scenario is not stationary, and the optimal solution changes with time. 

\begin{figure}
	\centering
	\includegraphics[scale=0.6]{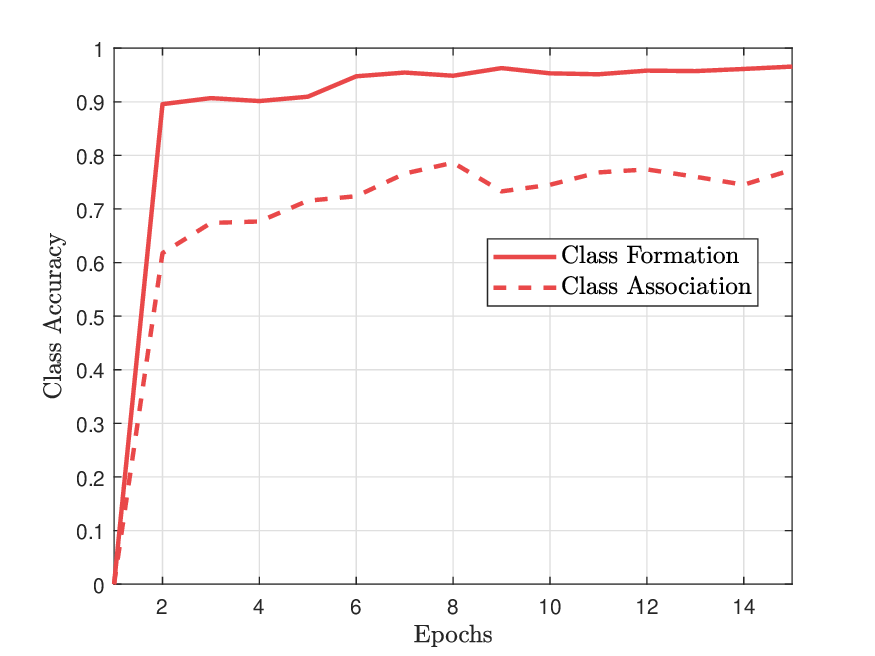}
	\caption{Class formation accuracy and track association accuracy. As the CC collects more target tracks, the accuracy of both the formation, updating, and association of classes improves. }
	\label{fig:classes}
\end{figure}

Figure \ref{fig:radar_utilization} shows that the mode control algorithm chooses about 80\% of the nodes for active observation in each time step. 
This is compared against a radar-only policy (which greedily selects only the radar option) and a ``random selection'' policy which selects radar 80\% of the time.

\begin{figure}
	\centering
	\includegraphics[scale=0.6]{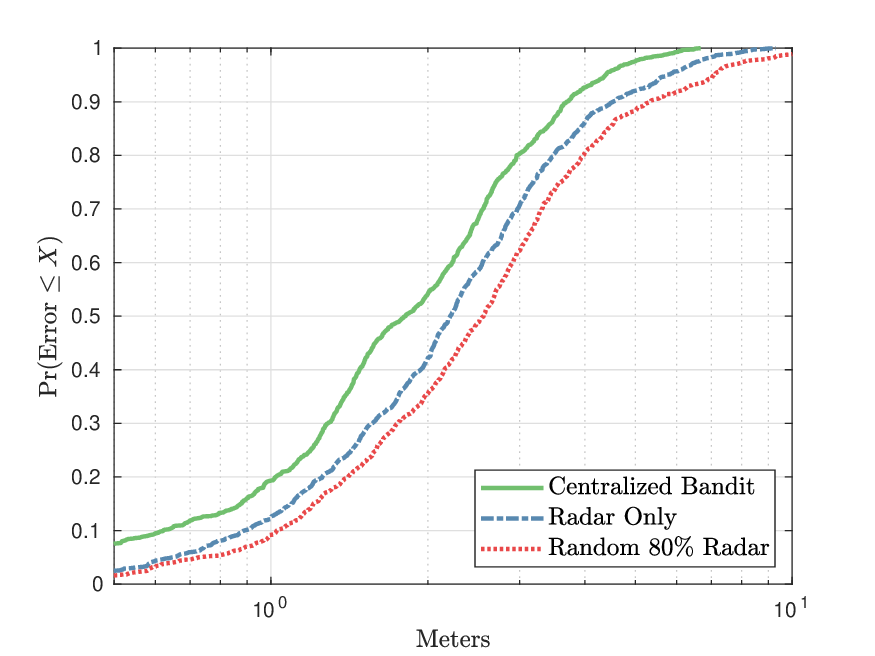}
	\caption{Utilization of mode control reduces the target tracking error due to better characterization of target motion. This represents the error distribution for the final epoch. }
	\label{fig:ecdf}
\end{figure}

\begin{figure}
	\centering
	\includegraphics[scale=0.6]{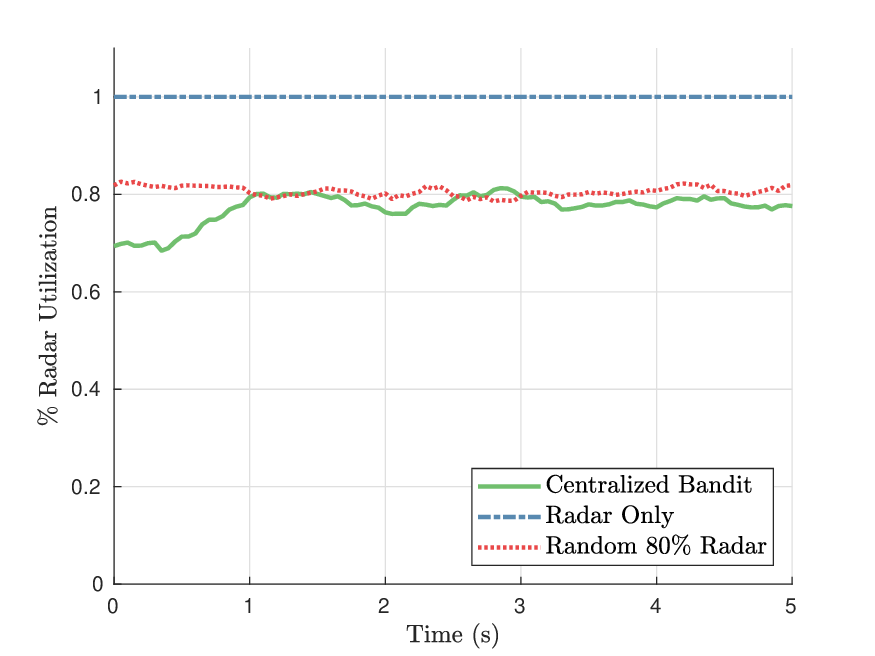}
	\caption{While the radar utilization does not change, the performance improves. }
	\label{fig:radar_utilization}
\end{figure}

Figure \ref{fig:ecdf} demonstrates that the centralized bandit mode selection policy obtains greater performance than either the 80\%-active policy (which uses radar the same portion of time) or the radar-only policy. 
The reported error is for the final epoch. 
The centralized bandit policy is better due to use of the passive action, which effectively characterizes targets so that the network is able to assign classes accurately.  
The random policy conducts radar 80\% of  the time, and obtains 10-15\% worse error than the bandit policy. 
The centralized bandit policy performs radar 70-80\% of the time and still obtains 5\% better error than the radar-only policy. 

Neither the radar-only policy nor the random policy utilize target class information. 
These results indicate that by forming target classes, a CRN is able to improve accuracy while reducing the need to perform active radar at all times. 
This in turn reduces the power consumption of the network as well as the maximum intercept range available to non-cooperative or adversarial systems which may try to impede the CRN.

%%%%%%%%%%%%%%%%%%%%%%%%%%%%%%%%%%%%%%%%%%%%%%%%%%%%%%%%%%%%%%%%%%%%%%%%%%%%%%%%%%%%%%%%%
%%%%%%%%%%%%%%%%%%%%%%%%%%%%%%%%%%%%%%%%%%%%%%%%%%%%%%%%%%%%%%%%%%%%%%%%%%%%%%%%%%%%%%%%%

\section{Conclusions}
\label{sec:conclusions}

In this work we investigated the capabilities available to a network when multiple modes of operation are present. 
This represents the first contribution towards the field of Cognitive Radar Network (CRN) mode control. 
In particular, we examine the case where the nodes of a CRN have, in addition to active radar, the ability to conduct passive signal parameter estimation. 
In each of many time steps, every node in the CRN can operate in one of these two modes. 
When conducting radar observation, a node provides to the Central Coordinator (CC) an estimate of the position and velocity of all targets within its range. 
When conducting passive signal parameter estimation, a node provides to the CC an estimate of the signal emissions from all targets within its range. 
The passive measurements are associated with targets via direction of arrival estimation. 

In addition to these direct observations, the CC maintains records of the motion model (i.e., constant velocity, constant turn, etc. ) and the history of signal emissions of each target. 
Modeling both of these as Markov processes, the CC estimates the transition probabilities for each of these parameters over time. 
On fixed intervals (``epochs’’), the CC then clusters these targets into ``classes’’ which contain targets with similar behavior. 
Finally, using these constructed target classes, the CC is able to estimate the class of future targets in order to determine their likely behavior. 
In this way, the motion model of targets is able to be estimated using \emph{passive} observation, as the class of a target is dependent on both signal and motion characteristics. 
So, the CC is able to trade between active radar observation and passive signal parameter estimation over time. 

We show that the use of this estimation technique can 1) reduce the effective radiated power of the network by decreasing the proportion of time radar is performed and 2) increase the target tracking accuracy of the network by leveraging ``prior'' information on the targets. 
We demonstrate that passive signal parameter estimation is constrained by the maximum detectable range and probability of intercept of the targets emissions.
% probability of detection, but when signals are detected, they are classified with near 100\% accuracy. 
In addition, we contribute mathematical analysis of the class formation technique. 

Future work in the field of CRN mode control will include ``distributed’’ control, where the nodes within the CRN choose for themselves the operation mode. 
One technique to accomplish this may be \emph{age of information} metrics, where the time since each action was performed is used to determine the next action, coupled with target estimation information.

%%%%%%%%%%%%%%%%%%%%%%%%%%%%%%%%%%%%%%%%%%%%%%%%%%%%%%%%%%%%%%%%%%%%%%%%%%%%%%%%%%%%%%%%%

% \clearpage
% \newpage
\bibliographystyle{IEEEtran}
\bibliography{bibli}

\end{document}